\documentclass[aps,
pra,reprint,superscriptaddress]{revtex4-2}
\usepackage{mathrsfs}
\usepackage{array}
\usepackage{booktabs}
\usepackage{tabu}
\usepackage{physics}
\usepackage{dcolumn}
\usepackage{amsmath}
\usepackage{amsfonts}
\usepackage{float}
\usepackage{amssymb}
\usepackage{amsthm}
\usepackage{graphicx,color}
\usepackage[colorlinks={true}]{hyperref}
\hypersetup{colorlinks=true,linkcolor=red,citecolor=blue,urlcolor=blue}
\usepackage{graphicx}
\usepackage{subfigure}
\usepackage{graphicx}
\usepackage{dcolumn}
\usepackage{bm}
\usepackage{pstricks}
\usepackage{braket}
\usepackage{orcidlink}

\def\be{\begin{equation}}
  \def\ee{\end{equation}}
\def\bea{\begin{eqnarray}}
\def\eea{\end{eqnarray}}

\newtheorem{theorem}{Theorem}
\newtheorem{corollary}[theorem]{Corollary}
\newtheorem{proposition}[theorem]{Proposition}

\theoremstyle{remark}
\newtheorem{remark}{Remark}

\newcommand{\Ecal}{\mathcal{E}}

\begin{document}

\title{High-dimensional coherence to entanglement transduction under canonical noise}

\author{Asad Ali~\!\!\orcidlink{0000-0001-9243-417X}}
\email{asal68826@hbku.edu.qa}
\affiliation{Qatar Center for Quantum Computing, College of Science and Engineering, Hamad Bin Khalifa University, Doha, Qatar}

\author{Aiham M. Rostom~\!\!\orcidlink{0000-0002-0084-0961}}
\affiliation{Institute of Automation and Electrometry SBRAS, 630090, Novosibirsk,
Russia}

\author{Saif Al-Kuwari~\!\!\orcidlink{0000-0002-4402-7710}}
\affiliation{Qatar Center for Quantum Computing, College of Science and Engineering, Hamad Bin Khalifa University, Doha, Qatar}

\author{H. Kuniyil~\!\!\orcidlink{0000-0003-0338-1278}}
\affiliation{Qatar Center for Quantum Computing, College of Science and Engineering, Hamad Bin Khalifa University,  Doha, Qatar}

\author{M.T. Rahim~\!\!\orcidlink{0000-0003-1529-928X}}
\affiliation{Qatar Center for Quantum Computing, College of Science and Engineering, Hamad Bin Khalifa University, Doha, Qatar}

\author{Saeed Haddadi~\!\!\orcidlink{0000-0002-1596-0763}}\email{haddadi@ipm.ir}
\address{School of Particles and Accelerators, Institute for Research in Fundamental Sciences (IPM), P.O. Box 19395-5531, Tehran, Iran}


\begin{abstract}
We develop an analytical framework for coherence-to-entanglement conversion in bipartite high-dimensional quantum systems, so-called qunits. An arbitrary coherent input qunit is coupled to an incoherent ancilla through a generalized controlled-shift operation, producing a maximally correlated bipartite state. By analyzing the partial transpose of the output state, we establish an exact dimension-independent connection between the input coherence and the generated entanglement. We then study how this conversion is affected by three standard noise processes applied after the conversion step: phase damping, global depolarizing noise, and independent amplitude damping. The resulting expressions show that these channels degrade entanglement in qualitatively different ways. Phase damping leads to a uniform attenuation of the entanglement generated from coherence, depolarizing noise introduces pairwise thresholds associated with entanglement sudden death, and amplitude damping produces an asymmetric decay governed by relaxation toward the ground state. For maximally coherent inputs, the general results reduce to simple closed-form behavior, allowing direct comparison of the three noise mechanisms as the system dimension increases. In particular, global depolarizing noise exhibits a dimension-dependent sudden-death threshold, while amplitude damping leads to a smooth suppression in the maximally coherent case. These results provide useful analytical benchmarks for high-dimensional resource conversion and for assessing noisy entanglement generation in qudit-based quantum-information settings.
\end{abstract}
\maketitle

\section{Introduction}
Quantum coherence and entanglement are two fundamental manifestations of nonclassicality. Coherence refers to the presence of off-diagonal matrix elements relative to a specified reference basis, while entanglement quantifies nonclassical correlations between distinct subsystems. Both resources underlie quantum computation, communication, sensing, thermodynamics, and many-body physics~\cite{NielsenChuang2000,Horodecki2009,PlenioVirmani2007,StreltsovRMP2017,ChitambarGour2019}. The modern resource theory of coherence formalizes basis-dependent superposition as a quantifiable resource and introduces coherence monotones such as the $l_1$-norm of coherence and relative entropy of coherence~\cite{Baumgratz2014,WinterYang2016,Napoli2016,Yuan2015,HU20181}. A particularly important observation is that local coherence can be converted into bipartite entanglement by suitable basis-preserving controlled operations~\cite{Streltsov2015,Ma2016,Killoran2016,Regula2018}.
Most elementary demonstrations of coherence-to-entanglement conversion are two-dimensional. However, many contemporary quantum platforms naturally support high-dimensional Hilbert spaces. Examples include photonic orbital angular momentum, path and time-bin encodings, multilevel atoms, trapped ions, superconducting circuits, and integrated photonic frequency combs~\cite{Durt2010,Erhard2020,Wang2020,Kues2017,Flamini2018}. High-dimensional quantum systems, often called qudits or qunits in dimension $n$, provide a larger information capacity per carrier, stronger noise tolerance in some communication protocols, and richer entanglement structures than two-level systems \cite{Cerf2002,Collins2002,Thew2004,AliKhan2007,Islam2017,Haddadi2025,Haddadi2025PRA}. This motivates a fully analytical treatment of the coherence-to-entanglement conversion in arbitrary finite dimensions.

The central object of this work is the controlled-shift conversion protocol. A coherent input qunit $A$ is prepared in an arbitrary pure state
\begin{equation}
\ket{\psi_A}=\sum_{j=0}^{n-1}c_j\ket{j}_A,
\end{equation}
while an ancillary qunit $B$ is initialized in the incoherent reference state $\ket{0}_B$. The generalized controlled-shift operation maps
\begin{equation}
\ket{j}_A\ket{k}_B
\mapsto
\ket{j}_A\ket{(k+j)\bmod n}_B.
\end{equation}
Consequently,
\begin{equation}
\ket{\psi_A}\ket{0}_B\mapsto \sum_{j=0}^{n-1}c_j\ket{jj}_{AB}.
\end{equation}
The output belongs to the class of maximally correlated states, whose entanglement properties are analytically tractable and deeply connected to coherence theory \cite{Rains2001,Vollbrecht2001,Streltsov2015,Regula2018}.

The first goal of this paper is to prove, in full generality, that the ideal controlled-shift protocol converts input $l_1$-norm of coherence $C_{l_1}(\rho_A)$ into output negativity $N_0$ according to the exact dimension-independent identity
\begin{equation}
N_0=\frac{1}{2}C_{l_1}(\rho_A).
\end{equation}
The second goal is to determine exactly how this conversion is degraded by physically important noise processes. We consider three canonical channels: independent phase damping, global depolarizing noise, and independent amplitude damping \cite{NielsenChuang2000,BreuerPetruccione2002,RivasHuelga2012,Wilde2017}. These channels capture, respectively, phase randomization without population relaxation, isotropic mixing with the maximally mixed state, and irreversible decay toward the ground level. Entanglement degradation under noise is a central theme in open quantum systems, including the phenomenon of entanglement sudden death \cite{YuEberly2004,YuEberly2009,Almeida2007,Ann2007}.
The main message is that the three noise mechanisms produce qualitatively different mathematical structures. Phase damping acts as a uniform multiplicative decay of every entanglement-carrying coherence. Global depolarization shifts each partial-transpose block by the same positive noise floor and hence produces pairwise sudden-death thresholds. Amplitude damping is more asymmetric: it treats the ground level differently from all excited levels, causing ground-excited coherence pairs to behave differently from excited-excited coherence pairs. These distinctions are invisible if one focuses only on the noiseless conversion identity.

The remainder of the paper is organized as follows. Section~\ref{sec:preliminaries} introduces notation, qunit states, coherence, and negativity. Section~\ref{sec:controlledshift} defines the controlled-shift conversion protocol. Section~\ref{sec:idealconversion} derives the ideal partial-transpose spectrum and proves the exact coherence-to-negativity identity. Section~\ref{sec:channels} defines the three noise channels and derives the exact noisy negativities. Section~\ref{sec:suddendeath} analyzes sudden-death thresholds and treats maximally coherent inputs, followed by the explicit low-dimensional examples. Section~\ref{sec:discussion} discusses physical interpretation and possible applications, and Sec.~\ref{sec:conclusion} concludes the article.

\section{Preliminaries}\label{sec:preliminaries}

\subsection{Qunit Hilbert spaces}

Let $A$ and $B$ be two $n$-level/dimensional quantum systems. Their Hilbert spaces are
\begin{equation}
\mathcal{H}_A\cong\mathbb{C}^n,
\qquad
\mathcal{H}_B\cong\mathbb{C}^n,
\end{equation}
with composite space
\begin{equation}
\mathcal{H}_{AB}=\mathcal{H}_A\otimes\mathcal{H}_B\cong\mathbb{C}^{n^2}.
\end{equation}
The computational basis of a single qunit is defined as $\{\ket{0},\ket{1},\ldots,\ket{n-1}\}$. So, the two-qunit basis is written as
\begin{equation}
\ket{jk}\equiv\ket{j}_A\otimes\ket{k}_B,
\qquad j,k\in\{0,1,\ldots,n-1\},
\end{equation}
and the basis satisfies
\begin{equation}
\langle j,k \mid l,m \rangle = \delta_{jl}\,\delta_{km},
\qquad
\sum_{j,k=0}^{n-1} \ket{j,k}\bra{j,k} = \mathbb{I}_{\mathbb{C}^n \otimes \mathbb{C}^n}.
\end{equation}

\subsection{Input state}

The input qunit $A$ is prepared in an arbitrary normalized pure state
\begin{equation}
\ket{\psi_A}=\sum_{j=0}^{n-1}c_j\ket{j}_A,
\qquad
\sum_{j=0}^{n-1}|c_j|^2=1,
\label{eq:input_state}
\end{equation}
and the complex amplitudes may be written as
\begin{equation}
c_j=\sqrt{q_j}e^{i\phi_j},
\qquad q_j\ge0,
\qquad \sum_{j=0}^{n-1}q_j=1.
\end{equation}
One global phase is physically irrelevant. The ancilla $B$ is initialized in the incoherent reference state $\ket{0}_B$. Thus, the initial two-qunit state is
\begin{equation}
\ket{\Psi_0}_{AB}=\ket{\psi_A}\otimes\ket{0}_B=\sum_{j=0}^{n-1}c_j\ket{j0}.
\label{eq:initial_twoqunit_state}
\end{equation}
This state is a product state and therefore contains no entanglement between $A$ and $B$.

\subsection{$l_1$-norm of coherence}

From Eq.~\eqref{eq:input_state}, the density operator of the input state is
\begin{equation}
\rho_A=\ket{\psi_A}\bra{\psi_A}=\sum_{j,k=0}^{n-1}c_jc_k^*\ket{j}\bra{k}.
\end{equation}
For this input state, the $l_1$-norm of coherence in the computational basis can be defined by \cite{Baumgratz2014}
\begin{equation}
C_{l_1}(\rho_A)=\sum_{j\ne k}|(\rho_A)_{jk}|.
\end{equation}
Since $(\rho_A)_{jk}=c_jc_k^*$, one obtains
\begin{equation}
C_{l_1}(\rho_A)=\sum_{j\ne k}|c_j||c_k|.
\label{eq:l1_sum}
\end{equation}
Using normalization, we have
\begin{align}
\left(\sum_{j=0}^{n-1}|c_j|\right)^2
&=\sum_{j=0}^{n-1}|c_j|^2+\sum_{j\ne k}|c_j||c_k|
=1+C_{l_1}(\rho_A),
\end{align}
so that
\begin{equation}
C_{l_1}(\rho_A)=\left(\sum_{j=0}^{n-1}|c_j|\right)^2-1.
\label{eq:l1_compact}
\end{equation}
Equivalently, we arrive at
\begin{equation}
C_{l_1}(\rho_A)=2\sum_{0\le j<k\le n-1}|c_jc_k|.
\label{eq:l1_pairwise}
\end{equation}
For more details, see Appendix~\ref{app:A}.

\subsection{Partial transpose and negativity}

For a density operator $\rho_{AB}$ on $\mathcal{H}_A\otimes\mathcal{H}_B$, the partial transpose with respect to subsystem $B$ is defined by
\begin{equation}
\matrixel{i,j}{\rho_{AB}^{T_B}}{k,l}=\matrixel{i,l}{\rho_{AB}}{k,j}.
\end{equation}
Equivalently,
\begin{equation}
\left(\ket{i}_A\ket{j}_B\bra{k}_A\bra{l}_B\right)^{T_B}=\ket{i}_A\ket{l}_B\bra{k}_A\bra{j}_B.
\end{equation}
Based on this definition, the negativity is introduced by \cite{VidalWerner2002,Plenio2005}
\begin{equation}
N(\rho_{AB})=\frac{\|\rho_{AB}^{T_B}\|_1-1}{2}.
\end{equation}
If $\{\lambda_r\}$ are the eigenvalues of $\rho_{AB}^{T_B}$, then
\begin{equation}
N(\rho_{AB})=\sum_{\lambda_r<0}|\lambda_r|.
\label{eq:negativity_negative_eigs}
\end{equation}
Note that for $2\otimes2$ and $2\otimes3$ systems, positivity of the partial transpose is necessary and sufficient for separability \cite{Peres1996,Horodecki1996}. In higher dimensions, PPT is not sufficient for separability, but negativity remains a computable entanglement monotone and detects non-PPT entanglement \cite{VidalWerner2002,Horodecki2009}.

\section{Controlled-Shift Coherence-to-Entanglement Conversion}\label{sec:controlledshift}

Let us define the generalized shift operator $X_n$ on $\mathbb{C}^n$ by
\begin{equation}
X_n\ket{k}=\ket{(k+1)\bmod n}.
\end{equation}
Its powers satisfy
\begin{equation}
X_n^j\ket{k}=\ket{(k+j) \bmod n}.
\end{equation}
The operator is unitary, with
\begin{equation}
X_n^\dagger=X_n^{-1}=X_n^{n-1},
\qquad
X_n^n=I_n.
\end{equation}
Also, the generalized controlled-shift operation is defined as
\begin{equation}
U_{\rm CS}^{(n)}=\sum_{j=0}^{n-1}\ket{j}\bra{j}_A\otimes X_n^j,
\label{eq:ucs_def}
\end{equation}
and its action on the computational basis is
\begin{equation}
U_{\rm CS}^{(n)}\ket{j}_A\ket{k}_B=\ket{j}_A\ket{(k+j) \bmod n}_B.
\label{eq:ucs_action}
\end{equation}
Unitarity follows from
\begin{align}
U_{\rm CS}^{(n)\dagger}U_{\rm CS}^{(n)}
&=\sum_{j,l=0}^{n-1}\delta_{jl}\ket{j}\bra{l}_A\otimes X_n^{-j}X_n^l\nonumber\\
&=\sum_{j=0}^{n-1}\ket{j}\bra{j}_A\otimes I_n=I_{n^2}.
\end{align}
Applying $U_{\rm CS}^{(n)}$ to Eq.~\eqref{eq:initial_twoqunit_state}, we obtain
\begin{align}
\ket{\Psi}_{AB}
&=U_{\rm CS}^{(n)}\sum_{j=0}^{n-1}c_j\ket{j0}
=\sum_{j=0}^{n-1}c_j\ket{jj}.
\label{eq:maxcorr_output}
\end{align}
So, the corresponding density operator is
\begin{equation}
\rho_{AB}=\ket{\Psi}_{AB}\bra{\Psi}=\sum_{j,k=0}^{n-1}c_jc_k^*\ket{jj}\bra{kk}.
\label{eq:maxcorr_density}
\end{equation}
This is a pure maximally correlated state. Its Schmidt coefficients are $\{|c_j|\}_{j=0}^{n-1}$, and its Schmidt rank is the number of nonzero amplitudes $c_j$.

\section{Ideal Conversion: Exact Partial-Transpose Spectrum}\label{sec:idealconversion}


Using $(\ket{jj}\bra{kk})^{T_B}=\ket{jk}\bra{kj}$, Eq.~\eqref{eq:maxcorr_density} gives
\begin{equation}
\rho_{AB}^{T_B}=\sum_{j,k=0}^{n-1}c_jc_k^*\ket{jk}\bra{kj}.
\label{eq:ideal_pt}
\end{equation}
This operator decomposes into one-dimensional diagonal blocks and two-dimensional pair blocks.
For $j=k$, Eq.~\eqref{eq:ideal_pt} contains $|c_j|^2\ket{jj}\bra{jj}$, so that each $\ket{jj}$ is an eigenvector with the following eigenvalue
\begin{equation}
\lambda_j=|c_j|^2\ge0.
\end{equation}
For every unordered pair $j<k$, the subspace $\mathcal{H}_{jk}={\rm span}\{\ket{jk},\ket{kj}\}$ is invariant under $\rho_{AB}^{T_B}$. In the ordered basis $\{\ket{jk},\ket{kj}\}$, the block is
\begin{equation}
B_{jk}^{(0)}=\begin{pmatrix}
0 & c_jc_k^*\\
c_j^*c_k & 0
\end{pmatrix}.
\end{equation}
The characteristic polynomial is $\det(B_{jk}^{(0)}-\lambda I)=\lambda^2-|c_jc_k|^2$, so the eigenvalues are (see Appendix~\ref{app:B})
\begin{equation}
\lambda_{jk}^{\pm}=\pm |c_jc_k|.
\label{eq:ideal_pair_eigs}
\end{equation}
The only negative eigenvalues are $-|c_jc_k|$ for $j<k$. Therefore,
\begin{equation}
N_0=\sum_{0\le j<k\le n-1}|c_jc_k|.
\label{eq:ideal_negativity}
\end{equation}
Using Eq.~\eqref{eq:l1_pairwise}, we obtain
\begin{equation}
N_0=\frac{1}{2}C_{l_1}(\rho_A).
\label{eq:ideal_identity}
\end{equation}

\begin{theorem}[Dimension-independent ideal conversion]
For the controlled-shift conversion protocol defined by Eq.~\eqref{eq:ucs_action}, with input state Eq.~\eqref{eq:input_state} and ancilla $\ket{0}_B$, the ideal output negativity is exactly one half of the input $l_1$-norm coherence:
\begin{equation}
N_0=\frac{1}{2}C_{l_1}(\rho_A).
\end{equation}
\end{theorem}

\begin{proof}
The partial transpose has non-negative diagonal eigenvalues $|c_j|^2$ and pairwise eigenvalues $\pm |c_jc_k|$ for all $j<k$. Hence, its negative spectrum contributes $\sum_{j<k}|c_jc_k|$ to the negativity. The $l_1$-coherence of the input is $2\sum_{j<k}|c_jc_k|$. Combining the two expressions gives the result.
\end{proof}

\section{Noise Channels}\label{sec:channels}

We now introduce the three noise models used throughout the paper. In all cases, the noise acts after the controlled-shift operation.

\subsection{Independent phase damping}

The phase-damping channel is defined phenomenologically by its action on the operator basis:
\begin{equation}
\ket{j}\bra{k}\longmapsto
\begin{cases}
\ket{j}\bra{j}, & j=k,\\
\sqrt{1-p}\ket{j}\bra{k}, & j\ne k.
\end{cases}
\label{eq:phase_single_action}
\end{equation}
The two-qunit channel is $\Ecal_{\rm ph}^{(2)}=\Ecal_{\rm ph}\otimes\Ecal_{\rm ph}$. A two-qunit coherence $\ket{jj}\bra{kk}$ with $j\ne k$ differs in both subsystems and is therefore multiplied by $\sqrt{1-p}\sqrt{1-p}=1-p$. The parameter $p\in[0,1]$ quantifies the strength of the phase damping:
$p=0$ leaves the system unchanged, while $p=1$ erases all off‑diagonal
coherences, reducing any state to a fully incoherent mixture in the
computational basis.

By applying Eq.~\eqref{eq:phase_single_action} to Eq.~\eqref{eq:maxcorr_density}, we obtain
\begin{equation}
\rho_{\rm ph}=\sum_{j=0}^{n-1}|c_j|^2\ket{jj}\bra{jj}
+(1-p)\sum_{j\ne k}c_jc_k^*\ket{jj}\bra{kk}.
\label{eq:rho_phase}
\end{equation}
The partial transpose is
\begin{equation}
\rho_{\rm ph}^{T_B}=\sum_{j=0}^{n-1}|c_j|^2\ket{jj}\bra{jj}
+(1-p)\sum_{j\ne k}c_jc_k^*\ket{jk}\bra{kj}.
\end{equation}
So, the diagonal eigenvalues are $|c_j|^2$. For each $j<k$, the pair block is
\begin{equation}
B_{jk}^{\rm ph}=\begin{pmatrix}
0 & (1-p)c_jc_k^*\\
(1-p)c_j^*c_k & 0
\end{pmatrix},
\end{equation}
with eigenvalues
\begin{equation}
\lambda_{jk,{\rm ph}}^{\pm}=\pm(1-p)|c_jc_k|.
\end{equation}
Therefore, we arrive at the following expression for negativity
\begin{equation}
N_{\rm ph}=(1-p)\sum_{j<k}|c_jc_k|=(1-p)N_0.
\label{eq:phase_negativity}
\end{equation}
Using Eq.~\eqref{eq:ideal_identity}, $N_{\rm ph}=\frac{1}{2}(1-p)C_{l_1}(\rho_A)$. The conversion efficiency relative to the ideal protocol is $\eta_{\rm ph}=N_{\rm ph}/N_0=1-p$.

\begin{proposition}[Universal phase-damping decay]
Independent phase damping suppresses the output negativity by a universal factor
$N_{\rm ph}=(1-p)N_0$.
This factor is independent of $n$ and independent of the probability distribution $\{|c_j|^2\}$.
\end{proposition}

\subsection{Global depolarizing channel}

The global two-qunit depolarizing channel is
\begin{equation}
\Ecal_{\rm dep}(\rho)=(1-p)\rho+\frac{p}{n^2}I_{n^2},
\label{eq:dep_channel}
\end{equation}
where $0\le p\le1$. This is a global channel on the full $n^2$-dimensional Hilbert space, not a product of local single-qunit depolarizing channels.
Applying Eq.~\eqref{eq:dep_channel} to the ideal state gives
\begin{equation}
\rho_{\rm dep}=(1-p)\sum_{j,k=0}^{n-1}c_jc_k^*\ket{jj}\bra{kk}+\frac{p}{n^2}I_{n^2}.
\label{eq:rho_dep}
\end{equation}
The partial transpose for this state can be written as
\begin{equation}
\rho_{\rm dep}^{T_B}=(1-p)\sum_{j,k=0}^{n-1}c_jc_k^*\ket{jk}\bra{kj}+\frac{p}{n^2}I_{n^2}.
\label{eq:rho_dep_pt}
\end{equation}
For $j=k$, the diagonal eigenvalues are
\begin{equation}
\lambda_j^{\rm dep}=(1-p)|c_j|^2+\frac{p}{n^2}\ge0,
\end{equation}
and for each $j<k$, the block in $\{\ket{jk},\ket{kj}\}$ is
\begin{equation}
B_{jk}^{\rm dep}=\begin{pmatrix}
\frac{p}{n^2} & (1-p)c_jc_k^*\\
(1-p)c_j^*c_k & \frac{p}{n^2}
\end{pmatrix}.
\end{equation}
Thus, the eigenvalues are
\begin{equation}
\lambda_{jk,{\rm dep}}^{\pm}=\frac{p}{n^2}\pm(1-p)|c_jc_k|.
\label{eq:dep_pair_eigs}
\end{equation}
Here, only $\lambda_{jk,{\rm dep}}^{-}$ can become negative. Therefore,
\begin{equation}
N_{\rm dep}=\sum_{j<k}\max\left\{0,(1-p)|c_jc_k|-\frac{p}{n^2}\right\}.
\label{eq:dep_negativity}
\end{equation}

\begin{proposition}[Depolarizing negativity]
For global two-qunit depolarizing noise, we have obtained an analytical expression for negativity as given by Eq.~\eqref{eq:dep_negativity}.
\end{proposition}

\subsection{Independent amplitude damping}
We use a zero-temperature amplitude-damping channel in which every excited state $\ket{a}$, $a\ge1$, relaxes to the ground state $\ket{0}$ with probability $p$. The single-qunit Kraus operators are
\begin{equation}
K_0=\ket{0}\bra{0}+\sqrt{1-p}\sum_{a=1}^{n-1}\ket{a}\bra{a},
\label{eq:ad_K0}
\end{equation}
\begin{equation}
K_a=\sqrt{p}\ket{0}\bra{a},
\qquad a=1,2,\ldots,n-1.
\label{eq:ad_Ka}
\end{equation}
They satisfy $K_0^\dagger K_0+\sum_{a=1}^{n-1}K_a^\dagger K_a=I_n$ (Appendix~\ref{app:C}). The two-qunit channel is $\Ecal_{\rm AD}^{(2)}=\Ecal_{\rm AD}\otimes\Ecal_{\rm AD}$.
Note that the amplitude damping channel requires a more detailed derivation because it distinguishes the ground state from the excited levels.


From the Kraus operators in Eqs.~\eqref{eq:ad_K0} and \eqref{eq:ad_Ka}, the relevant single-qunit transformations are
\begin{align}
\ket{0}\bra{0}&\longmapsto\ket{0}\bra{0},\nonumber\\
\ket{0}\bra{a}&\longmapsto\sqrt{1-p}\ket{0}\bra{a},\nonumber\\
\ket{a}\bra{0}&\longmapsto\sqrt{1-p}\ket{a}\bra{0},\\
\ket{a}\bra{a}&\longmapsto(1-p)\ket{a}\bra{a}+p\ket{0}\bra{0},\nonumber\\
\ket{a}\bra{b}&\longmapsto(1-p)\ket{a}\bra{b},\qquad a\ne b,
\quad a,b\ge1.\nonumber
\end{align}
These identities follow by direct substitution into $\Ecal_{\rm AD}(X)=\sum_{r=0}^{n-1}K_rXK_r^\dagger$, where $K_r$ denotes $K_0$ for $r=0$ and $K_a$ for $r=a\ge1$.


For amplitude-damped two-qunit state, we start from $\rho_{AB}=\sum_{j,k=0}^{n-1}c_jc_k^*\ket{jj}\bra{kk}$. Let's separate four types of terms. First, $\ket{00}\bra{00}\mapsto\ket{00}\bra{00}$. Second, for $a\ge1$,
\begin{align}
\ket{aa}\bra{aa}\longmapsto&(1-p)^2\ket{aa}\bra{aa}+p(1-p)\ket{a0}\bra{a0}
\nonumber\\
&+p(1-p)\ket{0a}\bra{0a}+p^2\ket{00}\bra{00}.
\label{eq:ad_diag_aa}
\end{align}
Third,
\begin{align}
&\ket{00}\bra{aa}\longmapsto(1-p)\ket{00}\bra{aa},
\nonumber\\
&\ket{aa}\bra{00}\longmapsto(1-p)\ket{aa}\bra{00}.
\label{eq:ad_00aa}
\end{align}
Fourth, for $a\ne b$ with $a,b\ge1$,
\begin{equation}
\ket{aa}\bra{bb}\longmapsto(1-p)^2\ket{aa}\bra{bb}.
\label{eq:ad_aabb}
\end{equation}
Combining Eqs.~\eqref{eq:ad_diag_aa}--\eqref{eq:ad_aabb}, the amplitude-damped density operator is given by
\begin{align}
\rho_{\rm AD}=&\left[|c_0|^2+p^2\sum_{a=1}^{n-1}|c_a|^2\right]\ket{00}\bra{00}
\nonumber\\
&+p(1-p)\sum_{a=1}^{n-1}|c_a|^2\left(\ket{a0}\bra{a0}+\ket{0a}\bra{0a}\right)
\nonumber\\
&+(1-p)^2\sum_{a=1}^{n-1}|c_a|^2\ket{aa}\bra{aa}
\nonumber\\
&+(1-p)\sum_{a=1}^{n-1}\left(c_0c_a^*\ket{00}\bra{aa}+c_ac_0^*\ket{aa}\bra{00}\right)
\nonumber\\
&+(1-p)^2\sum_{\substack{a,b=1\\a\ne b}}^{n-1}c_ac_b^*\ket{aa}\bra{bb}.
\label{eq:rho_ad}
\end{align}
The diagonal terms in Eq.~\eqref{eq:rho_ad} remain unchanged under partial transpose. The off-diagonal terms transform as
\begin{equation}
(\ket{00}\bra{aa})^{T_B}=\ket{0a}\bra{a0},
\qquad
(\ket{aa}\bra{00})^{T_B}=\ket{a0}\bra{0a},
\end{equation}
and $(\ket{aa}\bra{bb})^{T_B}=\ket{ab}\bra{ba}$. Therefore, $\rho_{\rm AD}^{T_B}$ decomposes into independent blocks of two types.

For each $a=1,\ldots,n-1$, the subspace $\mathcal{H}_{0a}^{\rm AD}={\rm span}\{\ket{0a},\ket{a0}\}$ contains the block
\begin{equation}
B_{0a}^{\rm AD}=\begin{pmatrix}
p(1-p)|c_a|^2 & (1-p)c_0c_a^*\\
(1-p)c_0^*c_a & p(1-p)|c_a|^2
\end{pmatrix},
\end{equation}
and its eigenvalues are
\begin{equation}
\lambda_{0a,{\rm AD}}^{\pm}=p(1-p)|c_a|^2\pm(1-p)|c_0c_a|.
\label{eq:ad_0a_eigs}
\end{equation}
For each pair $1\le a<b\le n-1$, the subspace $\mathcal{H}_{ab}^{\rm AD}={\rm span}\{\ket{ab},\ket{ba}\}$ contains
\begin{equation}
B_{ab}^{\rm AD}=\begin{pmatrix}
0 & (1-p)^2c_ac_b^*\\
(1-p)^2c_a^*c_b & 0
\end{pmatrix},
\end{equation}
with eigenvalues
\begin{equation}
\lambda_{ab,{\rm AD}}^{\pm}=\pm(1-p)^2|c_ac_b|.
\label{eq:ad_ab_eigs}
\end{equation}
All remaining eigenvalues are diagonal and non-negative.

From Eq.~\eqref{eq:ad_0a_eigs}, the negative part of the ground-excited block is $\max\{0,(1-p)(|c_0c_a|-p|c_a|^2)\}$. From Eq.~\eqref{eq:ad_ab_eigs}, each excited-excited block contributes $(1-p)^2|c_ac_b|$. Thus, we obtain the following expression for negativity as
\begin{align}
N_{\rm AD}=&\sum_{a=1}^{n-1}\max\left\{0,(1-p)\left(|c_0c_a|-p|c_a|^2\right)\right\}\nonumber\\
&+(1-p)^2\sum_{1\le a<b\le n-1}|c_ac_b|.
\label{eq:ad_negativity}
\end{align}

\begin{proposition}[Amplitude-damping negativity]
For independent zero-temperature amplitude damping on both qunits, we have obtained an analytical formula for negativity~\eqref{eq:ad_negativity}.
\end{proposition}

\section{Sudden-Death Thresholds and Maximally Coherent Inputs}\label{sec:suddendeath}


For depolarizing thresholds with the pair $(j,k)$, Eq.~\eqref{eq:dep_negativity} contributes positively if $(1-p)|c_jc_k|>p/n^2$. The threshold is obtained by equality:
\begin{equation}
(1-p_c^{(jk)})|c_jc_k|=\frac{p_c^{(jk)}}{n^2}.
\end{equation}
Solving, we arrive at
\begin{equation}
p_c^{(jk)}=\frac{n^2|c_jc_k|}{1+n^2|c_jc_k|}.
\label{eq:pc_dep_pair}
\end{equation}
The last surviving pair is the pair with maximal product $M=\max_{j<k}|c_jc_k|$. The global depolarizing threshold is therefore
\begin{equation}
p_c^{\rm dep}=\frac{n^2M}{1+n^2M}.
\label{eq:pc_dep_global}
\end{equation}


For amplitude-damping thresholds with a ground-excited pair $(0,a)$, the contribution is positive if $|c_0c_a|>p|c_a|^2$. If $|c_a|\ne0$, this becomes $p<|c_0|/|c_a|$. Thus a finite threshold in $[0,1)$ exists only when $|c_0|<|c_a|$, with
\begin{equation}
p_c^{(0a),\rm AD}=\frac{|c_0|}{|c_a|}.
\label{eq:pc_ad_0a}
\end{equation}
If $|c_0|\ge |c_a|$, the contribution vanishes only at $p=1$. For excited-excited pairs $(a,b)$ with $a,b\ge1$, the contribution is $(1-p)^2|c_ac_b|$, so it vanishes only at $p=1$ whenever $c_ac_b\ne0$.

\begin{remark}
The amplitude-damping channel distinguishes the computational basis state $\ket{0}$ as a physical ground level. Consequently, its sudden-death structure is not invariant under arbitrary relabeling of basis states. This is in contrast with the global depolarizing channel, which is isotropic on the full two-qunit Hilbert space.
\end{remark}


A maximally coherent input state has the form
\begin{equation}
\ket{\psi_A^{\rm mc}}=\frac{1}{\sqrt n}\sum_{j=0}^{n-1}e^{i\phi_j}\ket{j}_A.
\end{equation}
Then $|c_j|=1/\sqrt n$ for all $j$. There are $n(n-1)/2$ unordered pairs, each with $|c_jc_k|=1/n$. Therefore,
\begin{equation}
N_0^{\max}=\frac{n(n-1)}{2}\frac{1}{n}=\frac{n-1}{2}.
\label{eq:max_ideal}
\end{equation}
The input coherence is $C_{l_1}^{\max}=n-1$, consistent with $N_0^{\max}=C_{l_1}^{\max}/2$.

Hence, we obtain the following expression for phase damping
\begin{equation}
N_{\rm ph}^{\max}=\frac{n-1}{2}(1-p).
\label{eq:max_phase}
\end{equation}
Also, for global depolarization we have
\begin{align}
N_{\rm dep}^{\max}
&=\frac{n(n-1)}{2}\max\left\{0,\frac{1-p}{n}-\frac{p}{n^2}\right\}\nonumber\\
&=\frac{n-1}{2}\max\left\{0,1-p-\frac{p}{n}\right\},
\label{eq:max_dep}
\end{align}
which the depolarizing threshold is
\begin{equation}
p_c^{\max}=\frac{n}{n+1}.
\label{eq:max_pc_dep}
\end{equation}
For amplitude damping, the ground-excited contribution is $(n-1)(1-p)(1/n-p/n)=(n-1)(1-p)^2/n$, and the excited-excited contribution is $(1-p)^2(n-1)(n-2)/(2n)$. Adding both gives
\begin{equation}
N_{\rm AD}^{\max}=\frac{n-1}{2}(1-p)^2.
\label{eq:max_ad}
\end{equation}

\begin{corollary}[Maximally coherent input]
For $|c_j|=1/\sqrt n$, we get
\begin{align}
N_0^{\max}&=\frac{n-1}{2},\nonumber\\
N_{\rm ph}^{\max}&=\frac{n-1}{2}(1-p),\nonumber\\
N_{\rm dep}^{\max}&=\frac{n-1}{2}\max\left\{0,1-p-\frac{p}{n}\right\},\nonumber\\
N_{\rm AD}^{\max}&=\frac{n-1}{2}(1-p)^2,\nonumber
\end{align}
with depolarizing threshold as $p_c^{\max}=n/(n+1)$.
\end{corollary}

\begin{table*}[t]
\caption{
Summary of the principal analytical results for coherence-to-entanglement conversion in the proposed protocol.
}
\label{tab:mainresults}

\begin{tabular*}{\textwidth}{@{\extracolsep{\fill}}lll}
\hline\hline

Result &
General expression &
Maximally coherent state  \\
\hline

Input coherence &
$C_{l_1}=2\sum_{j<k}|c_j c_k|$ &
$C_{l_1}^{\mathrm{max}}=n-1$  \\

Ideal negativity &
$N_0=\sum_{j<k}|c_j c_k|
=\tfrac{1}{2}C_{l_1}$ &
$N_0^{\mathrm{max}}=\tfrac{n-1}{2}$  \\

Phase damping &
$N_{\mathrm{ph}}=(1-p)N_0$ &
$N_{\mathrm{ph}}^{\mathrm{max}}=\tfrac{n-1}{2}(1-p)$  \\

Global depolarization &
$N_{\mathrm{dep}}=\sum_{j<k}
\max\left[0,(1-p)|c_j c_k|-\tfrac{p}{n^2}\right]$ &
$N_{\mathrm{dep}}^{\mathrm{max}}=\tfrac{n-1}{2}
\max\left[0,1-p-\tfrac{p}{n}\right]$  \\

Critical dep. strength &
$p_c^{(jk)}=\frac{n^2|c_j c_k|}{1+n^2|c_j c_k|}$ &
$p_c^{\mathrm{max}}=\frac{n}{n+1}$  \\

Amplitude damping &
$N_{\mathrm{AD}}=
\sum_{a=1}^{n-1}
\max\left[0,(1-p)(|c_0 c_a|-p|c_a|^2)\right]+(1-p)^2\sum_{a<b}|c_a c_b|$ &
$N_{\mathrm{AD}}^{\mathrm{max}}=\tfrac{n-1}{2}(1-p)^2$  \\

Large-$n$ scaling &
$N_0^{\mathrm{max}}\sim n/2$ &
$p_c^{\mathrm{max}}\to 1$ \\

\hline\hline
\end{tabular*}
\end{table*}


For maximally coherent inputs, $N_0^{\max}=(n-1)/2$ grows linearly with $n$. The depolarizing threshold $p_c^{\max}=n/(n+1)$ approaches one as $n\to\infty$. The reason is that the depolarizing spectral shift scales as $p/n^2$, whereas each pairwise maximally coherent amplitude scales as $1/n$. Thus, the pairwise signal-to-noise ratio improves linearly with dimension under global depolarization.

Table~\ref{tab:mainresults} summarizes the central analytical results obtained throughout this work. The table highlights the exact proportionality between the generated entanglement and the initial quantum coherence in the noiseless protocol, together with the distinct degradation mechanisms induced by phase damping, depolarizing noise, and amplitude damping. While phase damping leads to a simple linear reduction of negativity, depolarizing noise introduces finite critical thresholds associated with entanglement sudden death. In contrast, amplitude damping produces an asymmetric decay arising from population relaxation toward the ground state. The table further reveals the favorable scaling of both negativity and noise robustness with increasing Hilbert-space dimension, demonstrating the advantages of high-dimensional coherent resources.



As an example, let us consider a qutrit case, $n=3$, as follows
\begin{equation}
\ket{\psi_A}=c_0\ket{0}+c_1\ket{1}+c_2\ket{2}.
\end{equation}
The controlled-shift output is
\begin{equation}
\ket{\Psi}=c_0\ket{00}+c_1\ket{11}+c_2\ket{22}.
\end{equation}
So, the ideal negativity is
\begin{equation}
N_0=|c_0c_1|+|c_0c_2|+|c_1c_2|.
\end{equation}
For a phase damping channel, the negativity takes the form
\begin{equation}
N_{\rm ph}=(1-p)(|c_0c_1|+|c_0c_2|+|c_1c_2|),
\end{equation}
and for global depolarization, we get
\begin{align}
N_{\rm dep}=&\max\left\{0,(1-p)|c_0c_1|-\frac{p}{9}\right\}\nonumber\\
&+\max\left\{0,(1-p)|c_0c_2|-\frac{p}{9}\right\}
\nonumber\\
&+\max\left\{0,(1-p)|c_1c_2|-\frac{p}{9}\right\}.
\end{align}
Besides, we obtain the following form for amplitude damping
\begin{align}
N_{\rm AD}=&\max\left\{0,(1-p)(|c_0c_1|-p|c_1|^2)\right\}
\nonumber\\
&+\max\left\{0,(1-p)(|c_0c_2|-p|c_2|^2)\right\}\nonumber\\
&+(1-p)^2|c_1c_2|.
\end{align}
Finally, for the maximally coherent qutrit input $|c_0|=|c_1|=|c_2|=1/\sqrt3$, one obtains
\begin{equation}
N_0^{\max}=1,
\qquad
N_{\rm ph}^{\max}=1-p,\nonumber
\end{equation}
\begin{equation}
N_{\rm dep}^{\max}=\max\left\{0,1-\frac{4p}{3}\right\},
\qquad
p_c^{\max}=\frac{3}{4},\nonumber
\end{equation}
\begin{equation}
N_{\rm AD}^{\max}=(1-p)^2.\nonumber
\end{equation}

\section{Discussion}\label{sec:discussion}

The exact results reveal a clear separation between three noise mechanisms. Phase damping preserves populations and uniformly suppresses all off-diagonal coherences. Since the entanglement generated by the controlled-shift operation is entirely carried by coherences of the form $\ket{jj}\bra{kk}$, the negativity simply inherits a universal factor $1-p$. No additional spectral offset appears in the partial transpose, and therefore no finite-$p$ sudden death occurs for any nonzero pairwise amplitude.

Global depolarization is qualitatively different. It suppresses the coherent part by $1-p$ and adds an isotropic identity contribution $pI_{n^2}/n^2$. In the partial transpose, the identity term uniformly lifts every pair block by $p/n^2$. A pair contribution survives only if the suppressed coherence amplitude exceeds this noise floor. This produces pair-resolved sudden-death thresholds governed by $|c_jc_k|$. The global entanglement threshold is controlled by the largest pairwise amplitude product, not by the total coherence alone.

Amplitude damping has a third structure. It is not isotropic and not merely dephasing; it selects $\ket{0}$ as a ground state. As a result, coherences involving $\ket{0}$ acquire diagonal populations in the partial-transpose blocks through decay of excited populations. These diagonal terms can eliminate the negative eigenvalue at finite $p$ when $|c_a|>|c_0|$. By contrast, coherences between two excited levels decay as $(1-p)^2$ and do not acquire the same diagonal lifting in the corresponding $\{\ket{ab},\ket{ba}\}$ block. Therefore, they vanish only at complete damping.

The maximally coherent case is especially transparent. In this case, all pairwise products are equal and the ideal negativity reaches $(n-1)/2$. The phase-damped negativity decays linearly in $1-p$, while the amplitude-damped negativity decays quadratically. The depolarizing threshold increases as $n/(n+1)$, showing a form of dimensional robustness against global white noise. This does not mean that arbitrary high-dimensional states are always more robust; rather, it means that the uniformly coherent state distributes pairwise coherence in a way that benefits from the $1/n^2$ scaling of the global depolarizing background.

These formulas may serve as analytic benchmarks for high-dimensional quantum-information tasks. Maximally correlated qunit/qudit states occur naturally in entanglement theory, quantum communication, dense coding, teleportation, and resource conversion \cite{Bennett1993,Bennett1992,Werner1989,Horodecki2009,Durt2010,Erhard2020}. Exact negativity expressions are valuable because high-dimensional entanglement measures are often difficult to compute. The present block-decomposition approach isolates the contribution of every pairwise coherence and makes the noise dependence transparent.

\section{Conclusion}\label{sec:conclusion}

We have presented a complete analytical theory of coherence-to-entanglement conversion in noisy two-qunit systems. A coherent input state $\ket{\psi_A}=\sum_jc_j\ket{j}_A$, together with an incoherent ancilla $\ket{0}_B$, is transformed by the generalized controlled-shift operation into the maximally correlated state $\sum_jc_j\ket{jj}_{AB}$. The partial transpose decomposes into independent two-dimensional pair blocks, yielding the exact ideal relation
$N_0=\frac{1}{2}C_{l_1}(\rho_A)$
for arbitrary dimension $n$.
We derived exact negativities under independent phase damping, global depolarizing noise, and independent zero-temperature amplitude damping. Phase damping gives a universal multiplicative law, $N_{\rm ph}=(1-p)N_0$. Global depolarization gives pairwise threshold contributions and a global sudden-death point determined by the largest product $|c_jc_k|$. Amplitude damping gives an asymmetric formula in which ground-excited and excited-excited coherence pairs have distinct decay laws.
For maximally coherent inputs, the noisy negativities reduce to simple closed forms, expressed in Corollary~5.
The depolarizing threshold $p_c^{\max}=n/(n+1)$ increases with dimension, while amplitude damping produces quadratic suppression rather than finite-$p$ sudden death in the maximally coherent case. These results provide closed-form benchmarks for high-dimensional resource conversion and clarify how different noise mechanisms limit the transformation of local coherence into bipartite entanglement.

\appendix

\section{Derivation of the $l_1$-norm of coherence identity}
\label{app:A}
For a pure state $\ket{\psi}=\sum_jc_j\ket{j}$,
\begin{equation}
\rho=\sum_{j,k}c_jc_k^*\ket{j}\bra{k}.
\end{equation}
Thus,
\begin{equation}
C_{l_1}=\sum_{j\ne k}|c_jc_k^*|=\sum_{j\ne k}|c_j||c_k|.
\end{equation}
The right-hand side counts each unordered pair twice:
\begin{equation}
\sum_{j\ne k}|c_j||c_k|=2\sum_{j<k}|c_j||c_k|.
\end{equation}
Since $|c_j||c_k|=|c_jc_k|$, Eq.~\eqref{eq:l1_pairwise} follows.

\section{Eigenvalues of two-dimensional Hermitian pair blocks}
\label{app:B}
The recurring block form is
\begin{equation}
B=\begin{pmatrix}
d & z\\
z^* & d
\end{pmatrix},
\end{equation}
where $d\in\mathbb{R}$ and $z\in\mathbb{C}$. The characteristic polynomial is
\begin{equation}
\det(B-\lambda I)=(d-\lambda)^2-|z|^2.
\end{equation}
Thus,
\begin{equation}
\lambda_{\pm}=d\pm |z|.
\end{equation}
For $d=0$, this gives $\pm|z|$. For depolarization, $d=p/n^2$. For the amplitude-damping ground-excited block, $d=p(1-p)|c_a|^2$ and $z=(1-p)c_0c_a^*$.

\section{Trace preservation of the amplitude-damping channel}
\label{app:C}
Using Eqs.~\eqref{eq:ad_K0} and \eqref{eq:ad_Ka}, we have
\begin{align}
&K_0^\dagger K_0=\ket{0}\bra{0}+(1-p)\sum_{a=1}^{n-1}\ket{a}\bra{a},\\
&\sum_{a=1}^{n-1}K_a^\dagger K_a=p\sum_{a=1}^{n-1}\ket{a}\bra{a}.
\end{align}
Adding gives
\begin{equation}
K_0^\dagger K_0+\sum_{a=1}^{n-1}K_a^\dagger K_a=\ket{0}\bra{0}+\sum_{a=1}^{n-1}\ket{a}\bra{a}=I_n.
\end{equation}
Therefore, the channel is trace preserving. Complete positivity follows from its Kraus representation.

\vspace{1cm}
\textbf{Author contributions:}
A.A. contributed to the conceptualization of the study, performed the calculations, interpreted the results, and prepared the original draft of the manuscript. A.M.R. and S.A. contributed to the conceptualization and investigation of the study and were involved in reviewing the manuscript. H.K. and M.T.R. contributed to the investigation and preparation of the original draft. S.H. contributed to the methodology and investigation and participated in reviewing and editing the manuscript.~All authors thoroughly checked the manuscript, discussed the results, and approved the final version of the manuscript.\\
\\
\textbf{Data~availability:}
All data that support the findings of this study are included within the article.\\
\\
\textbf{Competing~interests:}
The authors declare that they have no competing interests.\\
\\
\textbf{Research~funding:}
 This research did not receive funding.

\bibliography{Bibliography}

\end{document}